\documentclass[journal]{IEEEtran}
\usepackage {graphicx} % For � kunne inkludere grafikk
\usepackage {amsthm}
\usepackage[fleqn] {amsmath}
\usepackage {amsfonts, amssymb}
\usepackage{multirow}

\usepackage{enumerate}
\usepackage{color}
\usepackage{algorithmic} 
\usepackage[]{algorithm}
\newtheorem{theorem}{Theorem}

\newtheorem{proposition}[theorem]{Proposition}

\title{A Potential Game for Power and Frequency Allocation in Large-Scale Wireless Networks}
\author{Brage~Ellings\ae ter, Magnus~Skjegstad and Torleiv~Maseng}
\begin{document}
\maketitle
%\IEEEcompsoctitleabstractindextext{
\begin{abstract}
In this paper we analyze power and frequency allocation in wireless networks through potential games. Potential games are used frequently in the literature for this purpose due to their desirable properties, such as convergence and stability. However, potential games usually assume massive message passing to obtain the necessary neighbor information at each user to achieve these properties. In this paper we show an example of a game where we are able to characterize the necessary neighbor information in order to show that the game has a potential function and the properties of potential games. We consider a network consisting of local access points where the goal of each AP is to allocate power and frequency to achieve some SINR requirement. We use the physical SINR model to validate a successful allocation, and show that given a suitable payoff function the game emits a generalized ordinal potential function under the assumption of sufficient neighbor information. Through simulations we evaluate the performance of the proposed game on a large scale in relation to the amount of information at each AP.%Lastly we investigate the stability of these games in terms of Lyapunov stability and in light of comparative statics.
\end{abstract}

%}

%\maketitle

\section{Introduction}
With the demand for wireless capacity increasing exponentially, the concept of femtocells and other distributed networks are seen as potential technologies to remedy this demand. With these concepts comes the need to analyze and characterize these distributed networks. Game theory is a popular tool to investigate the behavior of autonomous agents as it is capable of characterizing convergence criteria and stability of different decision making policies. It has been used to analyze distributed wireless networks as early as \cite{Wu1995}\cite{Shah1998}. In most examples of applying game theory to wireless networks the problem statement is to optimize the network with regard to some parameter, e.g. maximize sum-rate, proportional fairness or average signal-to-interference-plus-noise-ratio (SINR), or minimize network interference. If these optimization problems satisfy some criteria (such as convexity) the overall optimization problem can be decoupled into individual optimization problems which can be solved individually by each user. However, due to the nature of interference this is rarely the case and the global optimization is usually non-convex and NP-hard.

Instead of proposing complex algorithms based on KKT conditions and second-order necessary conditions, one can propose simpler algorithms based on simplified system models \cite{Bensaou} or even common sense \cite{DaCosta2010} and use game theory to show properties such as convergence and stability. As convergence, stability and efficiency of steady-states depends on the chosen payoff function, designing a suitable payoff function is the most difficult part of modeling a system through game theory. Often small tweaks are necessary to achieve the desired properties of the game, thus it is important to keep in mind what the payoff function of the game reflects when investigating the efficiency of a steady-state.

One class of games where it is clear what the payoff function of each individual reflects when combined is potential games. The class of potential games consists of games which emits a potential function such that the increase or decrease in payoff of each individual under best-response play is reflected in a function. If this function reflects some overall system performance, it is easy to see how the decisions made by each individual affects the system as a whole. Potential games where introduced in \cite{Monderer1996}, where it was shown that this class of games exhibits very desirable properties. This include the existence of at least one pure strategy Nash equilibrium (NE), convergence to a NE and Lyapunov stability of the NEs (when the decision variable is continuous). Due to this, potential games have been used for resource allocation in decentralized wireless networks in \cite{Neel2004}\cite{Hicks2004}\cite{Comaniciu}\cite{di-li}. In addition, potential games in distributed networks is broadly covered in \cite{Neel2006}.

If one can find a potential function which reflects a desired aspect of system performance, then best-response play by the individuals will increase the system performance as well as their own payoff. The drawback is that to achieve this property each individual payoff has to embody the performance of the other users in some way. This is usually achieved through information exchange between users. The issue of limiting the information exchange to only include a subset of nodes was done in \cite{Canales2010} under the protocol interference model. In this model, interference from a transmitter is bounded in range so that interference from transmitters which are more than a threshold value away is neglected. However, in large scale networks the accumulated interference from all interfering transmitters can be much larger than that created by such a subset.% and we therefore use the physical SINR model  \cite{4146676} in this paper.

%Characterizing the number of necessary neighbors with which information has to be exchanged to preserve the properties of potential games has, to the best of the author's knowledge, not been published in the literature. In the literature this is either solved through the assumption of small networks or clusters, or neglected as a whole.

In this paper we characterize the necessary number of neighbors that each user needs to know the strategies of under the physical SINR model \cite{4146676}. Showing that such a subset exists under this interference model enables a distributed implementation, as global knowledge is not possible in a distributed approach. Combined with a practical neighbor discovery mechanism \cite{Skjegstad2012a} we show results from simulations based on the progress of neighbor discovery which validates that given sufficient information our game will converge.% and discuss the games in terms of Lyapunov stability and comparative statics.

\section{Preliminaries}
\label{sec:preliminaries}
\subsection{System Model}
We investigate a system consisting of $N$ wireless access points (APs). We envision these APs operate over unlicensed bands (either in the ISM bands or in the TV white spaces). The set of APs considered is denoted $\mathcal{N}$, where $|\mathcal{N}| = N$. 

To make our system model as generic as possible we let each AP have its individual spectral resources. In some systems, such as the ISM band, the spectral resources of each AP is the same, while in other systems, such as white space devices operating in the TV white spaces, the spectral resources depend on location. We therefore let the channel set available to each AP be a subset of the total channel set. The set of channels available to AP $i$ is denoted $\mathcal{K}_i$, with $|\mathcal{K}_i| = K_i$. The total set of channels is denoted $\mathcal{K}$, where $|\mathcal{K}| = K$, and we have that $\mathcal{K}_i\in\mathcal{K}$, $\forall i \in \mathcal{N}$.

We assume each AP has a maximum power $P_i^{\max}$ and a quality of service requirement. In this paper this quality of service metric is reflected through a SINR requirement $\beta_i$ for each AP $i$. The SINR of AP $i$ over channel $k$ is given as
\begin{equation}
%SINR_i(k) = \frac{g_{ii}(k)p_i(k)}{N_0 + \sum_{j\neq i}g_{ji}(k)p_j(k)} = \frac{g_{ii}(k)p_i(k)}{N_0 +I_i(k)}
SINR_i(k) = \frac{g_{ii}(k)p_i(k)}{N_0 +I_i(k)}
\end{equation}
where $p_i(k)$ is node $i$'s transmit power on channel $k$, $N_0$ is the power of the thermal noise and $I_i(k)$ is the interference observed at AP $i$ over channel $k$ and will be defined in the next subsection. $g_{ii}(k)$ is the channel gain from AP $i$ to its terminals. With multiple terminals, this can described as either the average channel gain between AP $i$ and its terminals, or as a worst case channel gain between AP $i$ and its terminals. This channel gain between AP $i$ and its terminals can be used to approximate a coverage area for AP $i$ with radius $r_i$. 

The goal of each AP is to select a channel such that its SINR requirement, $\beta_i$, is satisfied. This means that we find a channel $k^{*}$ such that $SINR_i(k^{*})\geq \beta_i$. To achieve its SINR requirement each AP is allowed to access \textit{one} channel and select a suitable power $p_i\leq P_i^{\max}$.  %We also assume each AP knows the SNR achievable from transmission on channel $k\in\mathcal{K}_i$, and can sense the interference power in each of its available channels. 

%Potential games have the desirable property that if each player adapts a best response strategy where only one player makes a move at each step, the game will converge to a NE. The NE obtained can also be shown to maximize the potential function of the game. Thus, if we can define the utility function of each player in such a manner that the game emits a potential function describing somehow the well-fair of the system, this well-fair function can be maximized by each player adopting a best-response strategy. If the utility function of each player only depends on local information, this can be achieved without any form of coordination.

\begin{figure}[t]
\centering
\includegraphics[width = 0.6\columnwidth]{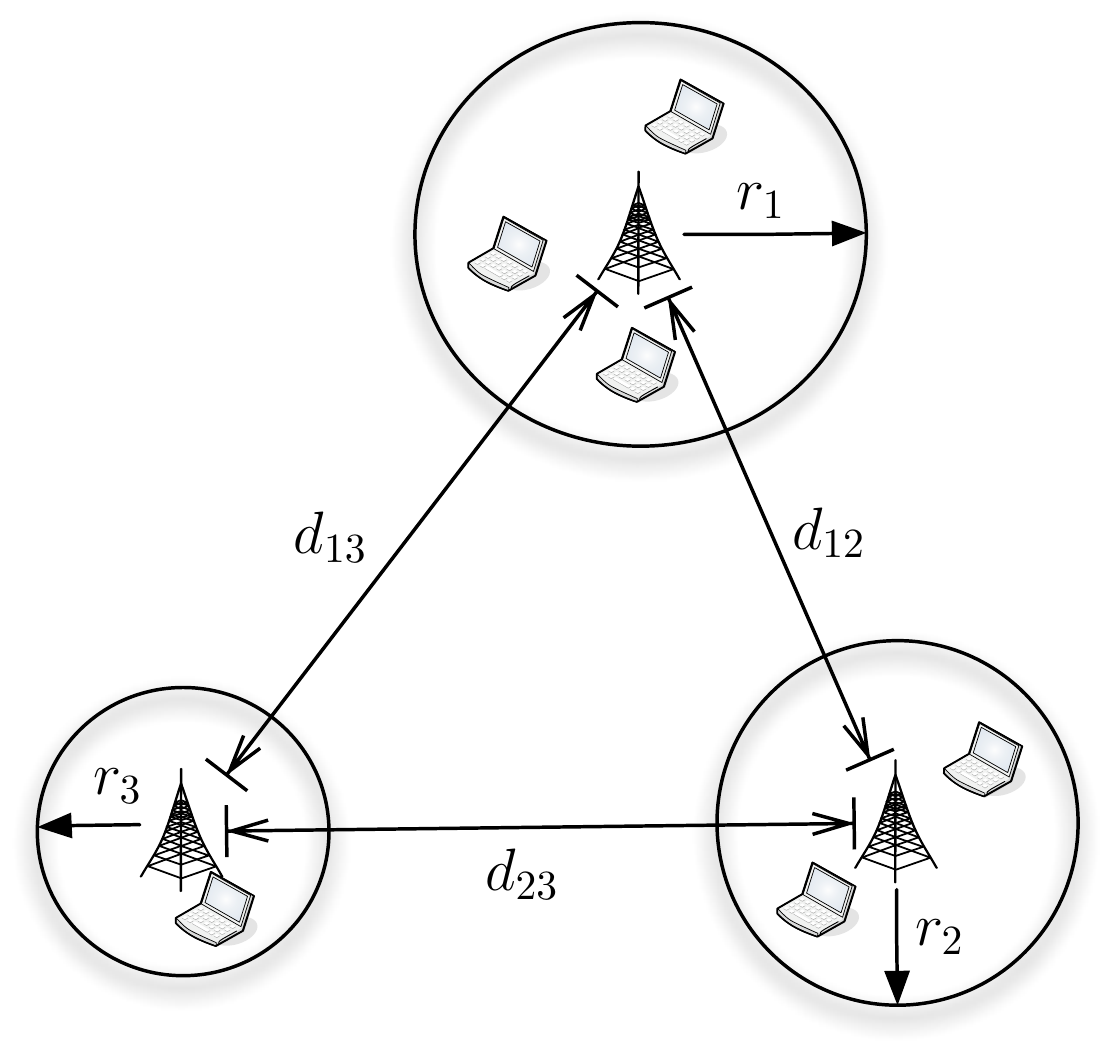}
\caption{Three access points with a distance $d_{ij}$ between them and coverage area given by $r_i, i = 1, 2, 3.$}
\label{fig:3-aps}
\end{figure}

\subsection{Interference Modeling}
%\subsection{Channel Gains}
In distributed networks it is difficult to characterize the interference between neighboring APs. Consider the simple network depicted in Fig. \ref{fig:3-aps}. The distance between each AP is given by $d_{ij} = d_{ji}$ and each AP has a coverage radius of $r_i$, $i= 1, 2,3$. The actual interference between the APs depends on the transmission mode of each AP, i.e. whether it is downlink or uplink transmission and also on the specific location of each terminal associated with each AP.

In \cite{sumrateMax_jasc11} an interference metric was proposed where the amount of interference created by AP $i$ to AP $j$ in the downlink was given as the average channel gain between AP $i$ and AP $j$'s terminals times AP $i$'s power in the downlink. If we assume channel gain is a product of distance dependent path loss and a stochastic variable, the average channel gain can written as
\begin{equation}
\hat{g}_{ij} = \frac{1}{|\mathcal{T}_j|}\sum_{l\in \mathcal{T}_j} d_{il}^{-\alpha}z_{il},
\end{equation}
where $\mathcal{T}_j$ is the set of AP $j$'s terminals, $d_{il}$ is the distance between AP $i$ and AP $j$'s terminal $l$, $\alpha$ is the path loss exponent and $z_{il}$ is a stochastic variable which can be distributed according to a fading distribution.

However, obtaining the necessary information from each terminal to justify this metric is a difficult process which assumes accurate sensing capabilities at the terminals and frequent updates to include the stochastic variable. Instead, we assume each AP has some sense of the average distance between itself and its terminals so that it can compute an estimate of its coverage area. The channel gain between AP $i$ and AP $j$ then becomes
\begin{equation}
g_{ij} = (d_{ij}-r_j)^{-\alpha}z_{ij}.
\end{equation}
We can simplify the channel gain further by the expectation of the stochastic variable
\begin{equation}
\bar{g}_{ij} = (d_{ij}-r_j)^{-\alpha}\mathbb{E}[z_{ij}] = (d_{ij}-r_j)^{-\alpha}\mu_z.
\end{equation}
A similar approach was done in \cite{di-li} to model channel gain and interference between APs and the validity of the model was verified through simulations. 
%\subsection{Accumuated Interference}
Having defined the channel gain between AP $i$ and AP $j$, the total interference observed at AP $j$ on channel $k$ is given as
\begin{equation}
I_j(k) = \sum_{i\in \mathcal{N}, i\neq j} g_{ij}p_i(k).
\end{equation}

%\subsection{Problem Description}
%The goal in this paper is to formulate a game between the players so that the number of players achieving their SINR requirement is maximized. This problem is NP-hard as can be seen as a non-convex (due to the interference coupling between APs) integer problem (since each AP can only select one channel). From classical optimization one can try to 
\subsection{Potential Games}
In game theory a game in normal form consists of a set of players $\Omega = \{1,...,N\}$, where each player $i$ has a strategy space $S_i\in S$ and an associated payoff function which gives each player a payoff $u_i(\mathbf{s}) = u_i(s_i,s_{-i})$ for a strategy profile $\mathbf{s}={s_1,...,s_N}$ with $s_1\in S_1, s_2\in S_2...$ \cite{fundenberg}. %All other players than player $i$ is usually referred to as $-i$.
We denote such a game as
\begin{equation}
\Gamma = \{\mathcal{N},S,\{u_i\} \}.
\end{equation}
%The game played in this paper is a dynamic sequential stage game. We denote an opening strategy profile as $s^{0}$, which is the initial choice made by all players. A stage is a period where each player makes one move. The players makes moves sequentially, i.e. first one player makes a move, then the next and so on. We denote stages by $E_p$, where $E_1$ is the stage after the opening strategy, $E_2$ is the stage after $E_1$ and so on.
In our game the players are the APs. APs and players are used interchangeably in the rest of this paper. The strategies of the players is the choice of the different available channels. Thus the strategy profile $\mathbf{s}$ consists of allocated channels, e.g. $\mathbf{s} = f_1,f_2,f_1,f_2$ if the game consists of 4 players and player one has chosen channel 1, player 2 has chosen channel 2, player 3 has chosen channel 1 and player 4 has chosen channel 2. %In a given stage a player can either either change its strategy from $s_i^{E_p}$ to $s_k^{E_{p+1}}$ or do nothing. Either way, this is considered a move by player $i$.
A Nash Equilibrium (NE) of the game is reached when no player has an incentive to unilaterally deviate from its current strategy. Thus, the strategy profile $\mathbf{s}^{*}$ is a NE, if only if
\begin{equation}
\forall i \in \mathcal{N}, u_i(s_i^{*},s_{-i}^{*})\geq u_i(s_i,s_{-i}^{*}), \forall s_i\in S_i\backslash s_i^{*}
\end{equation}

Let $P: \mathbf{s}\rightarrow \mathbb{R}$. We call $P$ a potential function of the game $\Gamma$ \cite{Monderer1996}. Specifically, we have that an exact potential game is characterized by , $\forall i\in \mathcal{N}$, for $s_i, s'_i \in S_i$
\begin{equation}
u_i(s_i,s_{-i}) - u_i(s'_i,s_{-i}) = P(s_i,s_{-i}) - P(s'_i,s_{-i})
\label{eq:exact-potential-game}
\end{equation}
where $P(\mathbf{s}) = P(s_i,s_{-i})$  is the potential function. An ordinal potential game is characterized by, $\forall i\in \Omega$, for $s_i, s'_i \in S_i$
\begin{equation}
u_i(s_i,s_{-i})>u_i(s'_i,s_{-i}) \Leftrightarrow P(s_i,s_{-i})>P(s'_i,s_{-i}).
\label{eq:ordinal-potential-game}
\end{equation}
Some important properties of potential games are summarized below:
\begin{proposition}{(Corollary 2.2 \cite{Monderer1996})}
Every finite ordinal potential game possesses a pure-strategy equilibrium.
\end{proposition}
\begin{proposition}{(Lemma 2.3 \cite{Monderer1996})}
Every finite ordinal potential game has the finite improvement path (FIP).
\end{proposition}
\begin{proposition}{(Theorem 5.26 \cite{Neel2006})}
Given a potential game $\Gamma = \{\Omega,S,\{u_i\} \}$ with potential function $P$, global maximizers of
$V$ are Nash equilibria.
\end{proposition}
%\begin{proposition}{(Theorem 5.35 \cite{Neel2006})}
%Let $S$ be a continuous action space and let $\mathbf{s}^{*}$ be an isolated maximizer of $P$. Then $\mathbf{s}^{*}$ is asymptotically Lyapunov stable.
%\end{proposition}

%Thus, potential games have at least one pure-strategy NE which can be obtained through a best-response dynamic being adopted by each player. One of these NEs are also maximizers of the potential function and under a continues action space these maximizers are Lyapunov stable. In an example involving frequency allocation, the action space is not continuous, but in many examples of power control this is the case. One down side is that if a game consists of multiple NEs, it is difficult to characterize which NE the game is converging to. In general this is dependent on the initial conditions. One way to overcome this problem is to try to design the payoff function such that only desirable strategy profiles are NEs.

Thus, potential games have at least one pure-strategy NE which can be obtained through a best-response dynamic being adopted by each player. One of these NEs are also maximizers of the potential function. One down side is that if a game consists of multiple NEs, it is difficult to characterize which NE the game is converging to. In general this is dependent on the initial conditions. One way to overcome this problem is to try to design the payoff function such that only desirable strategy profiles are NEs.

Although potential games have convergence due to the FIP and existence of a NE, the better response dynamics between players can not be arbitrary to achieve this convergence. We define round-robin timing as players updating actions according to a pre-determined order, random timing as one player updating its action at a time, but the order is random, asynchronous timing as a subset of players updating their actions at a time and synchronous timing as all players updating their actions simultaneously. Then we have the following property
\begin{proposition}{\cite{Neel2006}}
A potential game converges under round-robin, random and asynchronous timing. It does not converge under synchronous timing.
\end{proposition}
In practice this is not a limitation as the probability of all players choosing their action at the same time based on the same information is almost zero when these plauers are not synchronized.

\section{Potential Game Formulation}
%We now present a game model for the physical SINR model \cite{4146676}, which accounts for accumulated interference as opposed to graph-based models such as the one presented in the previous section. From a physical layer perspective this model is more accurate as it is the accumulated noise at the receiver that determines the SINR and BER and therefore the capacity of the user. In relation to the previous section, it is clear that if the previous game has a feasible solution, the game with cumulative interference has at least one feasible solution, and specifically one of these will be the solution to the previous game. Note also that if a matching of $M$ is possible in the previous game, a matching $M'\geq M$ is possible when we consider cumulative interference.

Due to the nature of interference maximizing the number of APs achieving their SINR requirements is NP-hard. This is easy to see as the problem is a non-convex (due to the interference coupling between APs) integer problem (since each AP can only select one channel). It is therefore unlikely that there exists a non-trivial potential game formulation where the potential function expresses the number of APs achieving their SINR requirements (a trivial formulation would be one involving global knowledge of all system parameters at each AP). Thus, instead of designing a potential function which maximizes the number of APs achieving their SINR requirements we use a potential function which tries to find a fair trade-off between the necessary transmit power at each AP and the interference each AP generates to neighboring APs.

%Let $\mathcal{T}$ be set of all transmit nodes and let $\mathcal{R}$ be the set of all receiver nodes, such that a user consists of one transmitter node and one receiver node. For simplicity we assume that transmit node $i$ wants to transmit to receiver node $i$. %We assume the nodes can choose one channel from a channel set $\mathcal{K}$ to transmit on. 
AP $i$'s necessary power to achieve its SINR requirement on channel $k$ is 
\begin{equation}
p_{i}^{\text{nec}}(k) \geq \min\Bigl(\frac{\beta_i(N_0+I_i(k))}{g_{ii}(k)}, P_i^{\max}\Bigr).
\label{eq:nec-power}
\end{equation}
Let the set $\mathcal{R}_i$ be the set of APs known by AP $i$, such that AP $i$ has an estimate of $g_{ij}$ and which channel each AP in $\mathcal{R}_i$ transmits on. %Note that from the P2P discovery mechanism $|\mathcal{R}_i|\leq M$. 
The utility function we propose is the following
\begin{align}
u_i(k) = &-\sum_{j\neq i,j\in\mathcal{N}}p_j(k)g_{ji}(k) \nonumber \\ &- \sum_{j\neq i, j\in \mathcal{R}_i}p^{\text{nec}}_i(k)\bar{g}_{ij}(k)l(p_j(k))
\label{eq:utility-function}
\end{align}
where $l(p_j(k)) = 1$ if $p_j(k)>0$ and $0$ otherwise. The first summation is the interference observed at AP $i$ and the second summation is the interference generated by AP $i$ at AP $i$'s known neighbors. We assume that the interference level can be measured and thus AP $i$ does not have to calculate the accumulated interference in each channel based on knowledge of the other transmitters. The measurement can either be done by the AP and the power of interference can be amplified to reflect the interference level at the edge of the coverage area, or the terminals can do measurements and report to the AP these measurements after which the AP calculates an average interference level for each channel. 

The transmit channel $k^{*}$ is then selected as
\begin{equation}
k^{*} = \max_{k\in \mathcal{K}} u_i(k)
\label{eq:k-star}
\end{equation}
and the power is set as
\begin{equation}
p_i(k) =  \left\lbrace\begin{array}{c c}{p^{\text{nec}}_i(k)} & {k=k^{*}}\\{0} & {\text{otherwise}}\end{array}\right.
\end{equation}

\subsection{Small Network}
Assume each AP $i$ knows the path loss between its self and all other APs in the network such that $\mathcal{R}_i = \mathcal{N}$ for all $i$. Then it was shown in \cite{Comaniciu} that
\begin{align}
P(\mathbf{s}) = \sum_{i\in \mathcal{N}}\bigl(&-\frac{1}{2}\sum_{j\neq i,j\in\mathcal{N}}p_j(k)g_{ji}(k) \nonumber \\ &-\frac{1}{2} \sum_{j\neq i, j\in \mathcal{N}}p^{\text{nec}}_i(k)\bar{g}_{ij}(k)l(p_j(k))\bigr)
\end{align}
is an exact potential function for the game with utility (\ref{eq:utility-function}).

\subsection{Large-scale Network}
In practice it is not possible to know how ones own transmission will affect all other APs in the network. Instead we try characterize the necessary neighbor APs AP $i$ needs to obtain information about to preserve the game as a potential game. We see from (\ref{eq:utility-function}) that an increase in $u_i$ due to power control can only occur when the transmit power at $i$ is decreased. Clearly this decrease of power at $i$ cannot decrease $u_j$ $\forall j\neq i$.

An increase in $u_i$ due to channel selection can occur in three ways:
\begin{enumerate}
\item increase due to the first sum in (\ref{eq:utility-function}): can change frequency to a channel with less interference, which leads to a decrease in power, which is in general good. However, this can increase the interference at some near-by APs.
\item increase due to the second sum in (\ref{eq:utility-function}): can change frequency to a channel which generates less interference. Good for near-by neighbors, but can result in an increase in power which can affect far-away APs.
\item both sums are decreased: this is generally beneficial.
\end{enumerate}
We see that an increase in $u_i$ can have various effects, except in the case where both sums are decreased. 

Let $\mathcal{N}_i$ be the set of neighbor APs closest to AP $i$ such that all channels are at least utilized by one AP in $\mathcal{N}_i$. Then we have the following result.
\begin{proposition}
The game formed by the utility function $u_i(k)$ converges if for all $i$ $\mathcal{N}_i\subseteq\mathcal{R}_i$.
\label{prop:resource-conv}
\end{proposition}
\begin{IEEEproof}
See Appendix \ref{app:proof-resource-conv}.
\end{IEEEproof}

\subsection{Scenarios where Selfish Converges}
The proposed utility function that each player wants to maximize is a trade-off between the interference observed by that player and the interference the player will create for the other players. A selfish approach would be to only choose a channel based on the interference observed by the player. However, analysis have shown that this approach does not converge in general \cite{di-li} and that the NEs reached by such selfish approach can be very inefficient \cite{Etkin2007}.

For our system model, there are cases where selfish behavior do converge. 
\begin{proposition}
Assume that each AP have the same coverage area, so that $r_i = r$, $\forall i \in \mathcal{N}$. Then a selfish best response dynamic by each player converges to a NE.
\label{prop:selfish}
\end{proposition}
\begin{IEEEproof}
See Appendix \ref{app:proof-propselfish}.
\end{IEEEproof}

Both Proposition \ref{prop:resource-conv} and Proposition \ref{prop:selfish} proves convergence by showing that the games are generalized ordinal potential games. We can simplify the scenario of Proposition \ref{prop:selfish} to the case where there is no power control and every AP transmits at the same maximum power $P$. In this case the game is a exact potential game.
\begin{proposition}
Assume each AP transmits at the same transmit power P and that there is no power control. Further assume that each AP have the same coverage area, so that $r_i = r$ $\forall i \in \mathcal{N}$. Then a selfish best response dynamic by each player converges to a NE.
\label{prop:selfish2}
\end{proposition}
\begin{IEEEproof}
%With no power control we can modify the utility function as follows
%\begin{equation}
%U_i(k) = &-P\sum_{j\neq i,j\in\mathcal{N}}g_{ji}(k)l(i,j)- P\sum_{j\neq i, j\in \mathcal{N}}g_{ij}(k)l(i,j)
%\end{equation}
With no power control and the same coverage radius for all APs the interference created by AP $i$ to AP $j$ is the same as interference created by AP $j$ to AP $i$. Thus $g_{ij} = g_{ji}$ and $I_{ij} = I_{ji}$. Thus the accumulated interference observed at AP $i$ is the same as the amount of interference AP $i$ will create at the other APs, and therefore the channel which has the least interference at AP $i$ is also the channel that will generate the least amount of interference at the other APs. In this case, selfish behavior forms an exact potential game which converges.
\end{IEEEproof}

\subsection{Local Optimality}
From Proposition \ref{prop:selfish2} we saw that if there is no power control and all APs have the same coverage area the channel with least interference at AP $i$ is also the channel on which AP $i$ will generate the least interference. Consider the general case, with power control and different coverage areas, if the channel with least interference at AP $i$ is also the channel on which AP $i$ will generate the least interference, how "optimal" is this choice of channel from a system performance perspective?

\begin{proposition}
If the selected transmit power at an AP is the channel with both the least amount of observed interference at that AP and channel which will generate the least amount of interference to the other APs, this is the optimal choice for the AP without global knowledge. 
\end{proposition}
\begin{proof}
Since the channel with least interference is the channel that requires least transmit power, this channel will lead to the least amount of generated interference from an AP to the rest of the set of APs if this channel also minimizes $\sum_{j\neq i, j\in \mathcal{N}} g_{ij}l(p_j(k))$. Thus if AP $i$ should be allowed to transmit, this would be the channel to select. However, from a system performance perspective, where the goal is to maximize the number of APs achieving their SINR requirement, there could be a case where the AP should choose a different channel. For instance, from a centralized view, the selection of the channel with least interference and generated interference can lead another AP to not achieve its SINR requirement, if this AP is operating at the limit of  maximum transmit power and interference it can tolerate, whereas choosing a non-optimal channel (which does not maximize (\ref{eq:utility-function}) would not do this. However, for AP to be able to determine this would require global information.
\end{proof}

\section{A Peer-to-Peer Discovery Mechanism for Large-scale Networks}
\label{sec:p2p}
How neighbors are discovered and how the information is exchanged does not matter for the resource allocation schemes presented in the two previous sections, as long as the information is obtained in some way. In this section we briefly describe a peer-to-peer discovery mechanism suited for neighbor discovery and information exchange that we have used in the simulation results \cite{Skjegstad2012a}.

\begin{figure}[t]
\centering
\includegraphics[width = 0.9\columnwidth]{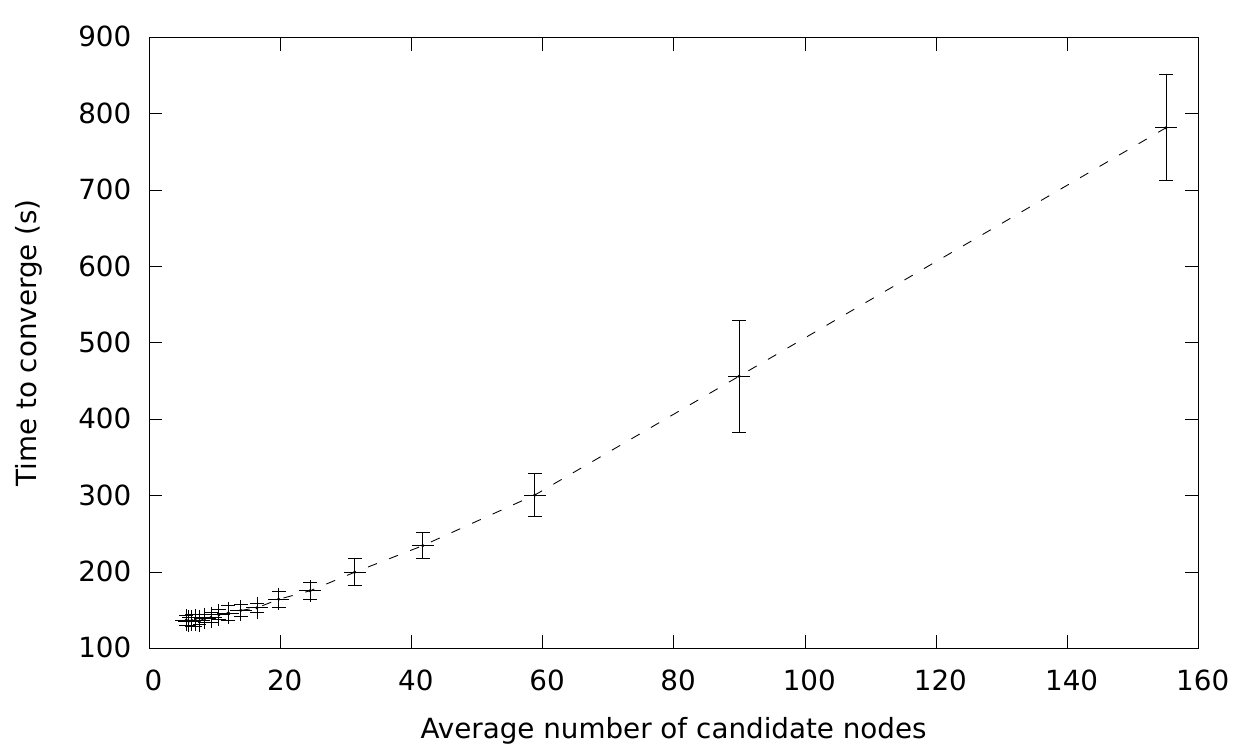}
\caption{Average time to reach stable state for varying neighbor node density in random topologies with standard deviation.}
\label{fig:discovery-convergence}
\end{figure}

The discovery mechanism consists of two parts: the first mechanism provides a random sample of all APs participating in the network. The second mechanism selects the most important APs from the random sample and exchanges information with them. As described in Section \ref{sec:preliminaries}, each AP has a coverage radius. We now assume that an AP has a coordination range, which is used to determine the importance of neighboring APs. Important neighboring APs are then determined by the extent to which the coordination areas overlap. When two nodes have overlapping coordination areas, they are defined as candidate nodes to each other. Note that this is just a rough estimate of whether two APs will interfere with each other.
\begin{table}
\caption{Evaluation parameters for the resource allocation algorithm.}
	\centering
	\begin{tabular}{c|c}
	\hline
		Parameter & Value \\ \hline \hline
		Number of Channels ($|\mathcal{K}|$) & 13 \\ \hline
		Maximum transmit power & 100 mW \\ \hline
		SINR requirement & $\sim U(1,6)$ \\ \hline
		Thermal noise power & $10^{-8}$ W \\ \hline
		Path loss between node $i$ and $j$, $g_{i,j}$ & $d(i,j)^{-3}\times z_{ij}$ \\ \hline
		%Path loss exponent, $\alpha$ & 3 \\ \hline
		%Shadowing $z$ & 
		Number of APs considered & 305 \\ \hline
		Maximum coverage radius ($d_{\text{transmit}}$) & 20,40 m\\ \hline
		Maximum number of iterations & 50 \\ \hline
	\end{tabular}
	\label{tab:RAparams}	
\end{table}
The medium that convey the messages of the discovery mechanism is not specified in \cite{Skjegstad2012a}. However, assuming the transmitters are access points or other terminals that provide Internet connectivity, one option is that the peer-to-peer network operate over the Internet. When a node connects to the P2P network, it begins to sample the network randomly. As it obtains knowledge about other nodes in the network, it exchanges information with these nodes while still sampling the network randomly. As time progresses, more and more candidate nodes are discovered until all APs know about their candidate nodes. Fig. \ref{fig:discovery-convergence} shows the average time it takes all APs to discover all their candidate nodes with standard deviation for increasing network sizes (or increasing number of neighbors). Although it may take some time to converge, convergence is guaranteed given enough time.

\section{Performance Evaluation}
\subsection{Simulation Setup}
To evaluate the resource allocation algorithm we have simulated a topology of APs based on the municipality of Tynset in Norway. We use household locations as the basis for AP location. We consider an area of $1\times 1$ km which consists of 305 APs. We set $|\mathcal{K}| = 13$ to mimic the number of available channels in 802.11b/g/n. Unlike Wi-Fi, we assume all channels are orthogonal. Each AP has a maximum transmit power of 100 mW, which is the maximum EIRP for Wi-Fi in Europe. It is also defined as the maximum transmit power for mobile white space devices in the US \cite{fcc4}. The signal-to-interference-plus-noise ratio (SINR) requirement for each AP is uniformly distributed between $1$ and $10$. The coverage area of each AP is uniformly distributed between $3$ and $d_{\text{transmit}}$, where $d_{\text{transmit}}$ is either 25 or 50. The coordination range used by the discovery mechanism is set to two times the maximum coverage radius. Thermal noise power is set to $10^{-8}$. The path loss between a transmitter $i$ and receiver $j$, $g_{ij}(k)$, is given in our simulations as $g_{i,j} = d(i,j)^{-3}\times z_{ij}$, where $z_{ij}$ is distributed according to a log-normal distribution with 0 dB mean and 8 dB standard deviation. The simulation parameters are summarized in Table \ref{tab:RAparams}. Note that even though we have included shadowing in the simulations, knowledge of the shadowing gain between between AP $i$ and AP $j$ is assumed unknown in the estimation of the last sum in (\ref{eq:utility-function}). Finally, an AP is defined as satisfied if its SINR at its coverage edge satisfies the SINR target. 

We evaluate the resource allocation algorithm in light of both its quality as a resource allocation algorithm and how it progresses based on neighbor discovery from the discovery mechanism in Section \ref{sec:p2p}. As very little information is known at the start-up of a network, the resource allocation algorithm can not make correct decisions. However, as time progresses, the discovery mechanism discovers more and more neighbors and the performance of the proposed resource allocation algorithm increases. For all plots, the results are plotted against time as the discovery protocol finds more candidate nodes (relevant neighboring APs). Once a candidate node is found, the candidate node exchanges location knowledge along with its estimated coverage radius. It will also continuously send information about its current transmit channel.

\begin{figure}[t!]
\centering
\includegraphics[width = 1.02\columnwidth]{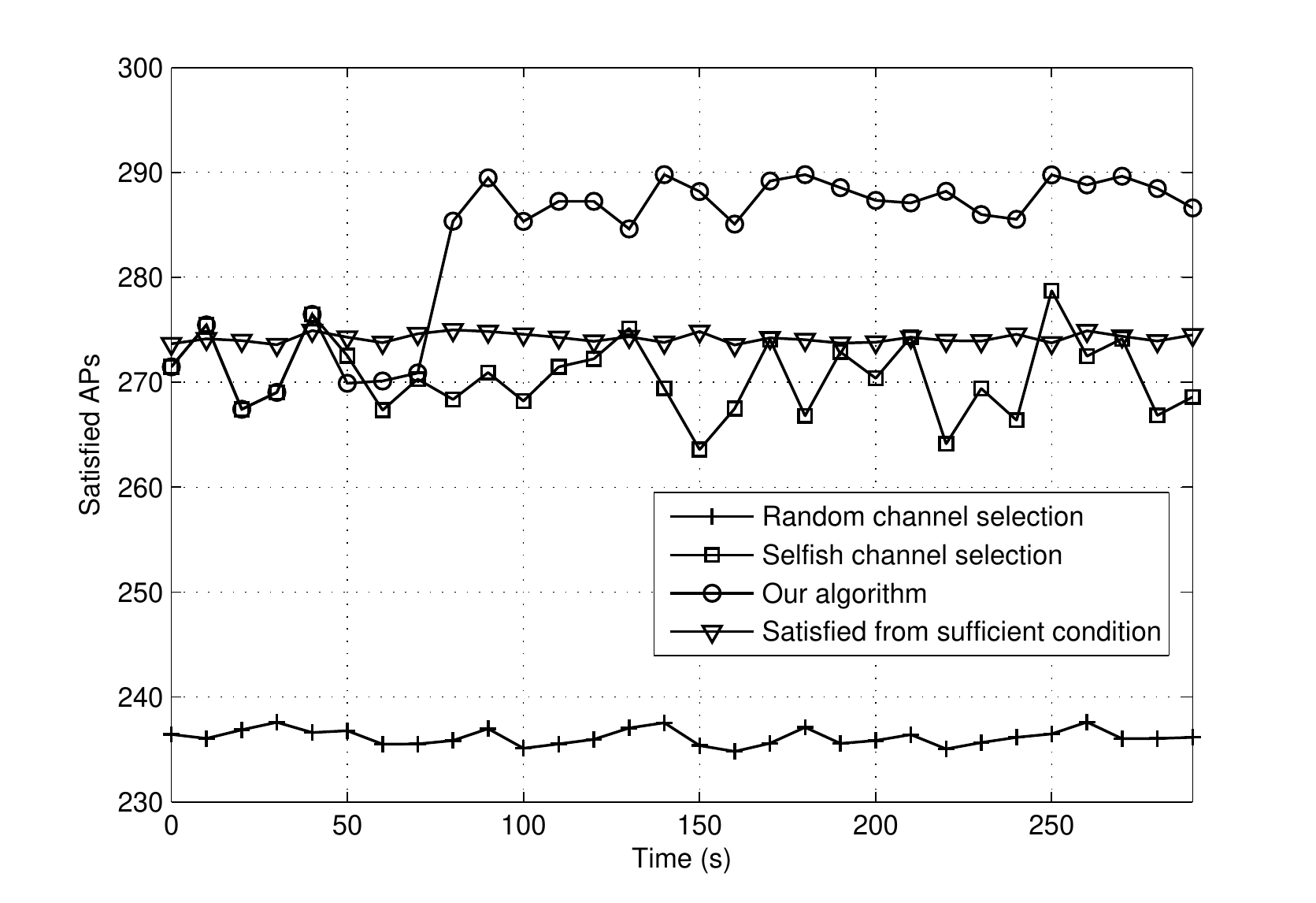}
\caption{Number of satisfied APs as a function of time in Tynset with maximum coverage radius of 20 meters.}
\label{fig:res-sat-tyn-r1-25-r2-50}
\end{figure}

The resource allocation algorithm is compared to two simple allocation schemes and centralized solution. The two allocation schemes are \emph{random} and \emph{selfish}.The random allocation scheme selects a frequency randomly from the set of available channels. The selfish scheme selects the channel with least interference and transmits with a power just large enough to achieve its SINR requirement. Both our proposed algorithm and the selfish algorithm starts with a random allocation of channels.

As finding the maximum number of users that achieve their SINR is NP-hard, comparing our algorithm to an optimal solution is not possible. However, in \cite{Ellingsater2012:a} a simple sufficient condition is shown to exists for which an efficient heuristic algorithm can find the solution. Thus, the solution value from this heuristic algorithm gives a lower bound on the number of APs that can theoretically achieve their SINR simultaneously. The drawback with this approach is that it assumes global knowledge of all system parameters and requires APs not achieving their SINR to power off.

\begin{figure}[t!]
\centering
\includegraphics[width = 0.7\columnwidth, angle=90]{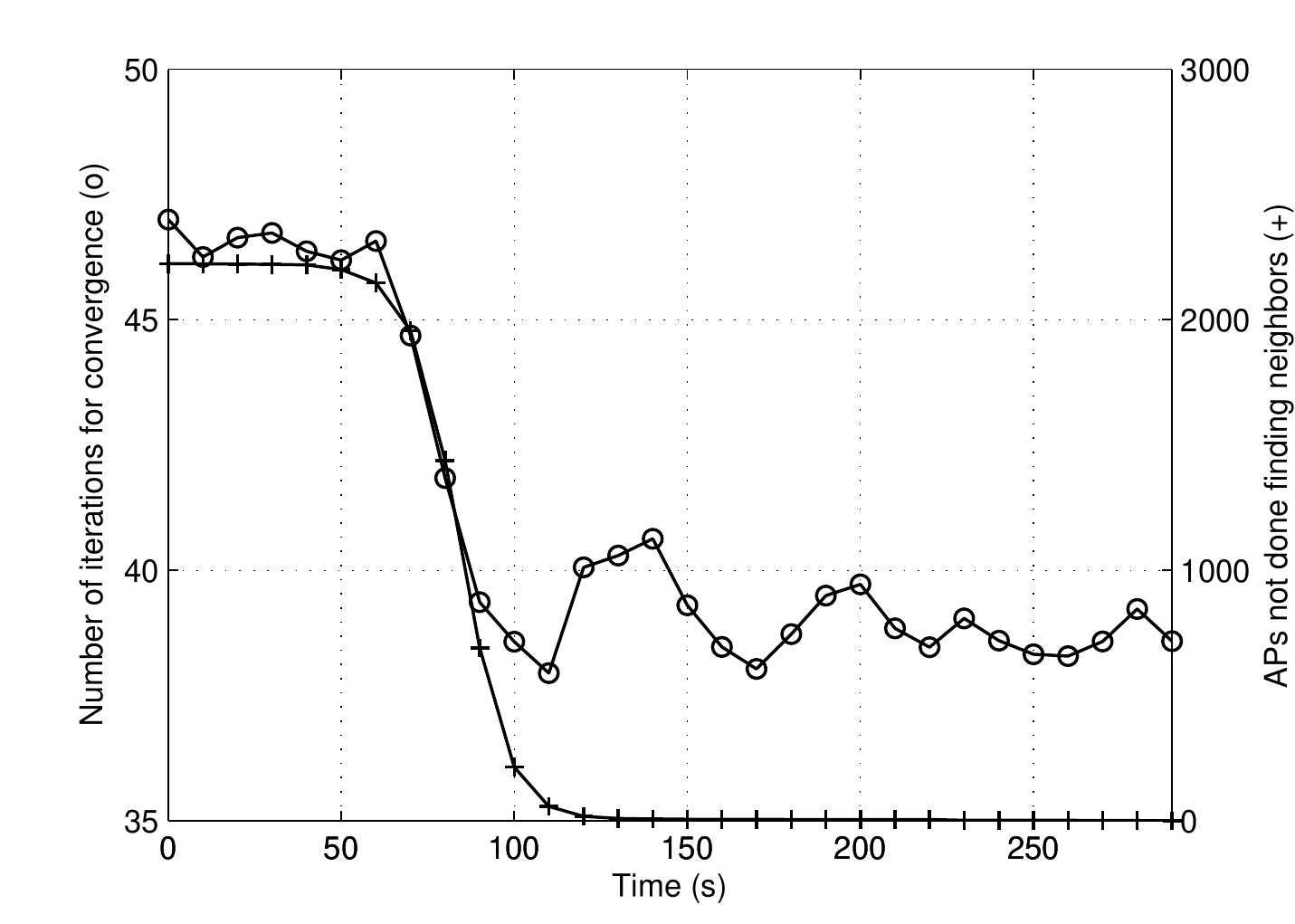}
\caption{Number of iterations until convergence against time in Tynset with maximum coverage radius of 20 meters}
\label{fig:res-conv-tyn-r1-25-r2-50}
\end{figure}
%The total number of APs in the Tynset topology is 2913 (total number of nodes is 5826). As allocating resources for all 2913 transmitters would require enormous processing, we only evaluate the resource allocation algorithm based on an area of 1 x 1 km. In this area there are 305 APs to consider. %The APs are created from household locations as described in Section \ref{sec:topology}. Note that the maximum distance between two linked nodes is half the coordination range.

%, with a uniformly distributed distance between transmitters and receivers from 1 to $d_{\text{transmit}}$. In this paper we have varied considered two cases: $d_{\text{transmit}} = 25$ m. and $d_{\text{transmit}} = 50$ m. The coordination radius of each node is given as two times the distance between the transmitter and receiver.

%We begin by assigning a random channel to each of the 2913 transmitters and transmit power equal to 100 mW. We then start the algorithm for the transmitters within the 1 x 1 km area. If nodes have receiver neighbors outside the 1 x 1 km grid, these are still considered in the frequency allocation utility function (Equation \ref{eq:utility-function}) of the transmitters. 

%\textbf{Brage: Er det utility func til p2p eller allokeringsmekanismen?}
\subsection{Simulation Results}
In Fig. \ref{fig:res-sat-tyn-r1-25-r2-50}, the number of satisfied APs is given in Tynset with a maximum coverage radius of 20 meters. As expected, the random allocation and selfish scheme behave the same over time, as their performance is not related to neighbor discovery. The same holds for the centralized solution. For our algorithm we see a clear increase in the number of satisfied APs between 50 to 100 seconds. The increase has two dependent causes: 1) From 50 to 100 seconds the discovery mechanism goes from a state where very few APs know all their neighbors to a state where most APs know all their neighbors, as seen in Fig. \ref{fig:res-conv-tyn-r1-25-r2-50}. 2) From Fig. \ref{fig:res-conv-tyn-r1-25-r2-50} we also see that our algorithm goes from a high average of 47 to 38 iterations. The average of 47 is close to the maximum number of iterations (50) we set for the algorithm and is due to the algorithm not being able to converge when too little information is available about neighbors. When more information becomes available the algorithm converges in less than 50 iterations for most instances and the average goes down. Note that as some APs are without neighbors there are fewer nodes with missing neighbors than the total number of APs in Fig. \ref{fig:res-conv-tyn-r1-25-r2-50}.

In essence, this shows that in order for our algorithm to perform well, the algorithm must converge. For our algorithm to converge, we have to know a sufficiently amount of neighbors such that at least one neighbor utilizes each channel (by Proposition \ref{prop:resource-conv}). This seems to occur at about 100 seconds.

\begin{figure}[t!]
\centering
\includegraphics[width = 1.02\columnwidth]{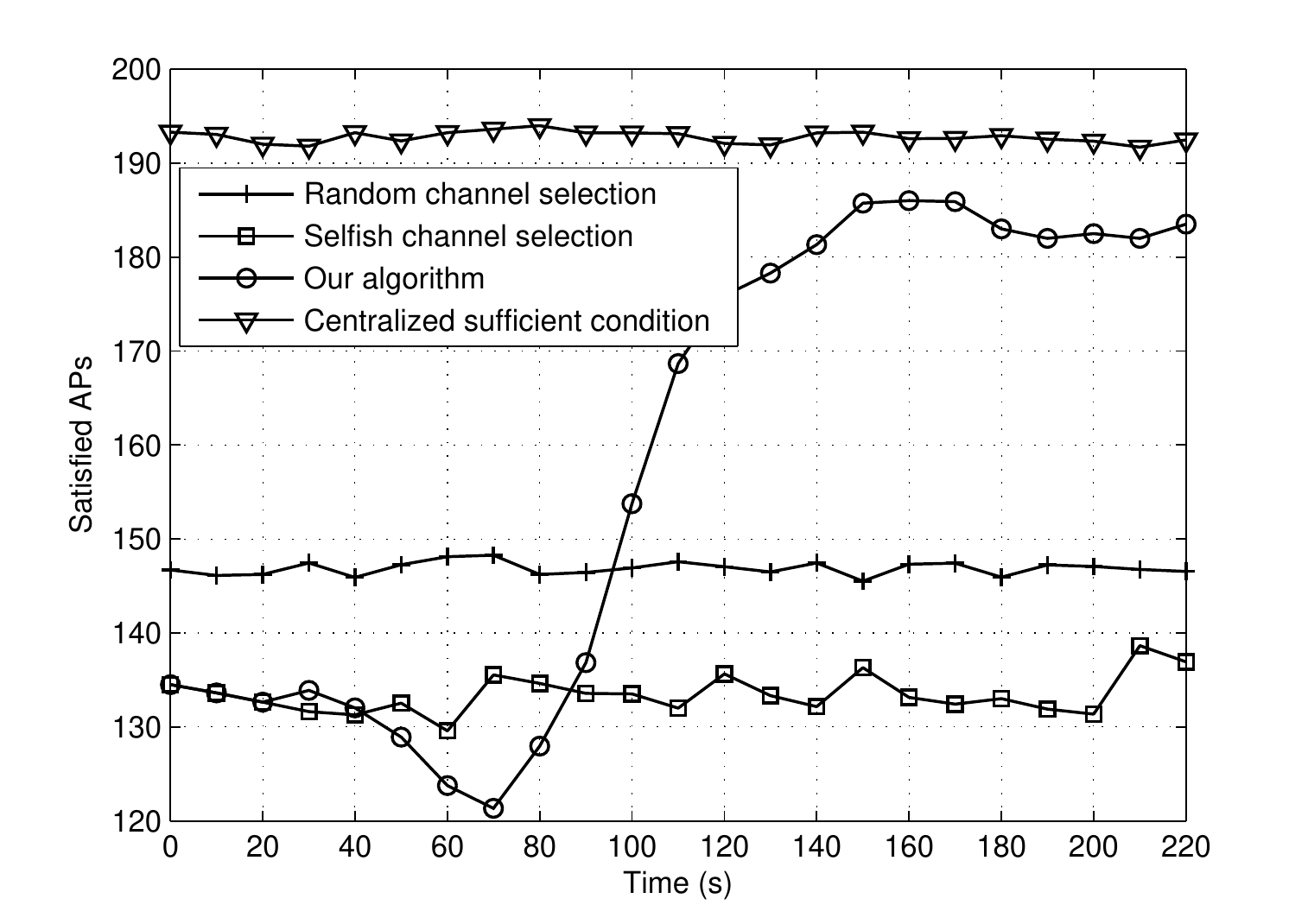}
\caption{Number of satisfied APs as a function of time in Tynset with maximum coverage radius of 40 m.}
\label{fig:res-sat-tyn-r1-50-r2-100}
\end{figure}

The importance of convergence is further emphasized in Fig. \ref{fig:res-sat-tyn-r1-50-r2-100} and Fig. \ref{fig:res-conv-tyn-r1-50-r2-100} where the maximum coordination range is increased to 80 meters and the maximum coverage radius is 40 meters. As some APs do not have candidate APs, the number of initial APs with missing candidate APs is also here lower than the total number of APs. Our algorithm starts off with performance similar to the selfish approach. At 50 seconds the performance drops below the selfish selection scheme, before it starts to improve and settles at around 150 seconds. The convergence plot in Fig. \ref{fig:res-conv-tyn-r1-50-r2-100} is similar to Fig. \ref{fig:res-conv-tyn-r1-25-r2-50} and we can again conclude that the algorithm performs well when the algorithm is able to converge. 

The drop in performance from 50 to 70 seconds can be explained as follows: from 50 to 70 seconds some APs find some of their relevant neighboring APs. However, due to the large coverage radius, and thus large coordination radius, some of these candidate nodes can be far away, while other more important candidate nodes are not found. Thus the utility function (\ref{eq:utility-function}) takes the well-being of far away neighbors into account, but not the ones close by. This drop in performance did not occur in Fig. \ref{fig:res-sat-tyn-r1-25-r2-50}, due to the small coordination radius.

In Fig. \ref{fig:res-sat-tyn-r1-25-r2-50} our algorithm performed better than the centralized solution. As the centralized solution is obtained from a sufficient condition on optimality it is not an upper bound and can surpassed. However, in  Fig. \ref{fig:res-sat-tyn-r1-50-r2-100} our algorithm does not surpass the centralized solution, but comes close as discovery mechanism discovers more neighbors. An explanation for why algorithm does not achieve the same or better as the sufficient condition is that in this scenario a lot fewer APs can achieve their SINR requirements, since the maximum coverage radius is increased from 20 to 40 m. To achieve the global optimum in this case some APs must turn of their power in order for APs to achieve their SINR targets. As our approach does not provide any access control all APs are allowed to transmit and the therefore the global optimum cannot be achieved by our algorithm in this case.
\begin{figure}[t!]
\centering
\includegraphics[width = 0.7\columnwidth, angle=90]{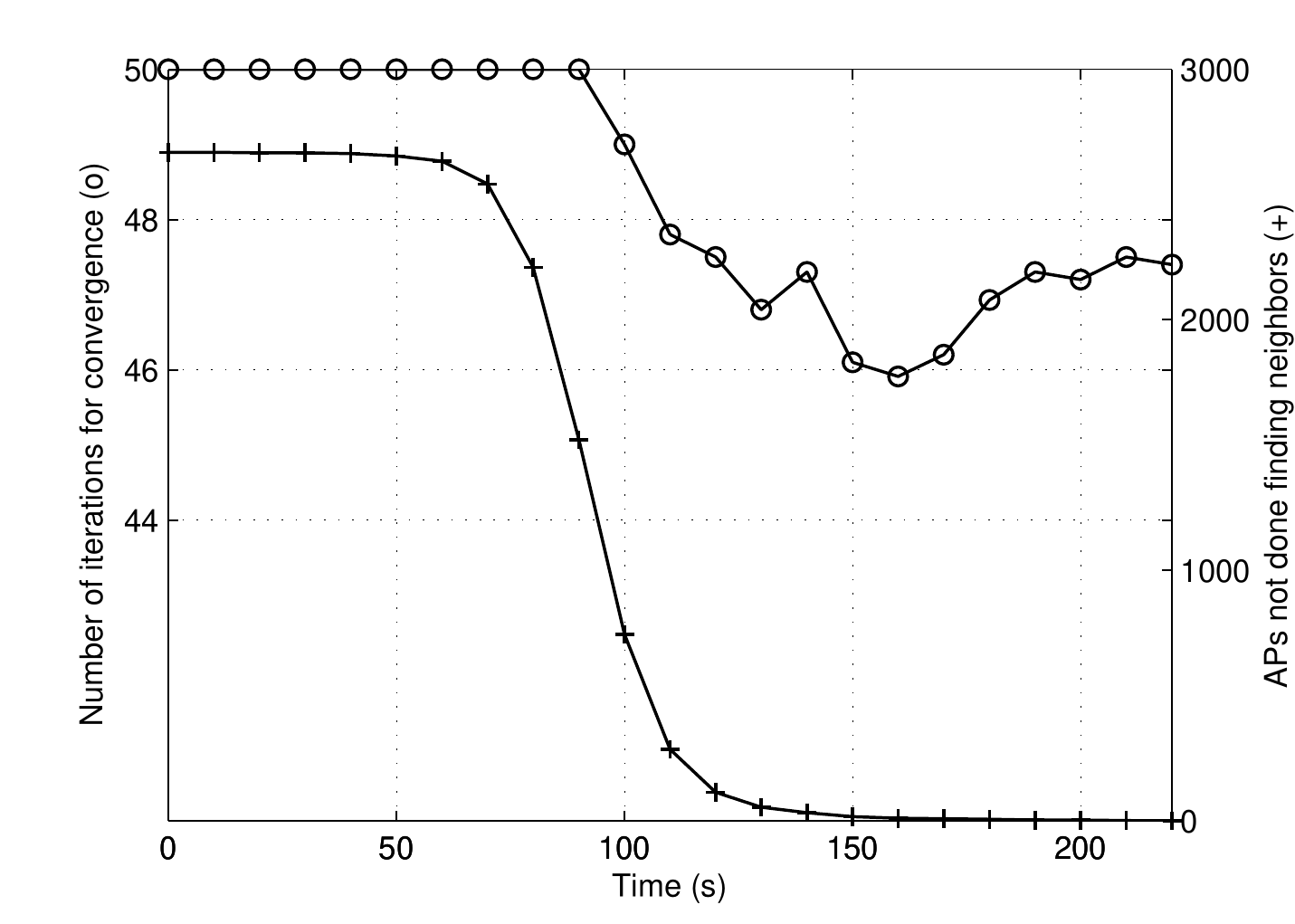}
\caption{Number of iterations until convergence against time in Tynset with maximum coverage radius of 40 m.}
\label{fig:res-conv-tyn-r1-50-r2-100}
\end{figure}
An important aspect of distributed resource allocation schemes in large scale networks is the \textit{domino effect} \cite{Ellingsater2012}, which is the change in the network due to a change in resource allocation at one or more APs. It is desirable that this effect is as small as possible, i.e. as few APs as possible should have to update their allocation due to a change at another AP. In general, changing power levels at a transmitter only leads to a minor domino effect as a large change in one node affects the other nodes in a much smaller scale as the path loss exponent is larger than 2. Predicting the effect of frequency change is more difficult.

To investigate this effect we let our original network converge (220 seconds in Fig. \ref{fig:random-inserted-nodes}) and then insert 10 APs in our 1 x 1 km area (230 seconds in Fig. \ref{fig:random-inserted-nodes}). We measure the change over time as the number of APs that change their frequency from one time interval to the next. From Fig. \ref{fig:random-inserted-nodes} we see that approximately one third of the APs change their frequency at the time the new APs are inserted into the network. After this point only 2-3 APs change their frequency for each 10 second interval until after 290 seconds where 13 APs change their frequency. After this, the change goes down to just above 0. An interesting observation is that when the 13 APs change their frequency at 290 seconds, the number of satisfied APs also increases. It seems that this change occurs from the discovery of new near-by neighbors both by new and old APs. When these APs adjust their frequency assignments the number of satisfied APs increases.

\begin{figure}[t!]
\centering
\includegraphics[width = 1.02\columnwidth]{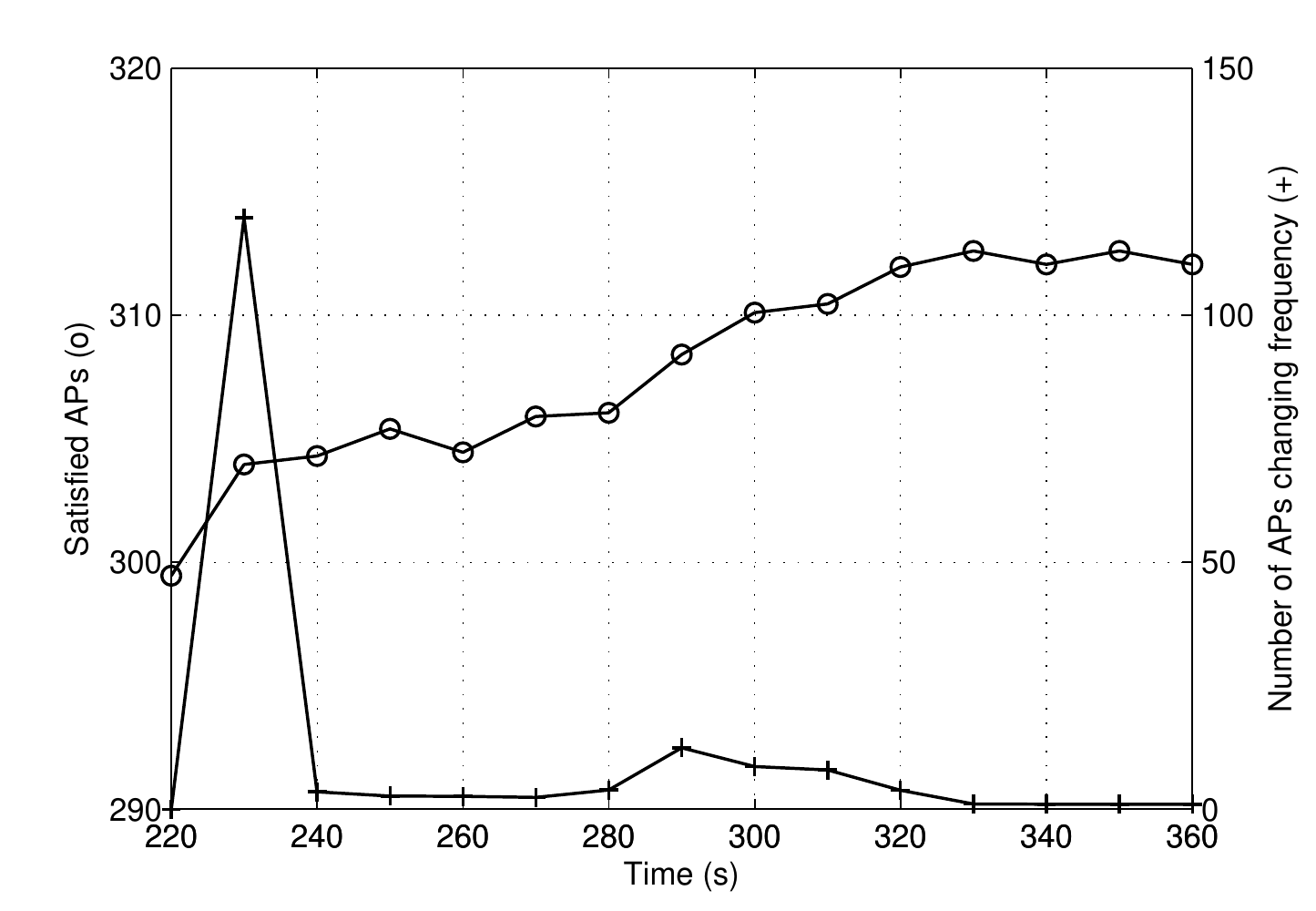}
\caption{Satisfied APs and number of APs changing frequency as a function of time. 10 additional APs are added to a converged network at 230 seconds.}
\label{fig:random-inserted-nodes}
\end{figure}

\section{Conclusion}
In this paper we investigated power and frequency allocation in a large-scale wireless network through potential games. As potential games have desirable properties such as pure strategy Nash equilibria and guaranteed convergence. For these reasons potential games have been used to analyze wireless networks in the past. Our main result is the characterization of the necessary amount of information each AP needs about its neighbors in order to conserve the potential game property of our proposed game. We have also shown that in certain cases the best response dynamic of each AP according to the proposed game is the optimal strategy of each AP, in the absence of global information at each AP. To obtain the necessary information we used a P2P discovery mechanism proposed in a previous paper. Through simulations we showed that the game performs better than current best-practice in such large scale networks with regard to performance, convergence and robustness, but this performance is dependent on neighbor information as suggested by our analytical analysis.

\appendices
\section{Proof of Proposition \ref{prop:resource-conv}}
\label{app:proof-resource-conv}
We prove convergence by showing that the game forms a generalized ordinal potential game. Convergence is then guaranteed as all finite generalized ordinal potential games have a pure strategy equilibrium \cite[Corollary 2.2]{Monderer1996} and all finite generalized ordinal potential games have a finite improvement path \cite[Lemma 2.3]{Monderer1996}.

A game with utility function $u_i$ constitutes an ordinal potential game if there exists a potential function such that $u_i(s_i,s_{-i})>u_i(s'_i,s_{-i}) \Rightarrow P(s_i,s_{-i})>P(s'_i,s_{-i})$.

From the definition of $\mathcal{N}_i$ we can rewrite the utility function as
\begin{align}
u_i(k) = &-\sum_{j\neq i,j\in\mathcal{N}_i}p_j(k)g_{ji}(k) - \sum_{j\neq i,j\in\mathcal{N}\backslash\mathcal{N}_i}p_j(k)g_{ji}(k) \nonumber \\ &- \sum_{j\neq i, j\in \mathcal{R}_i}p^{\text{nec}}_i(k)g_{ij}(k)f(p_j(k)).
\label{eq:new-utility}
\end{align}
where the first term is the observed interference at AP $i$ due to known neighboring APs of $i$, the second term is interference due to non-neighboring APs of $i$, and the third term is interference generated by $i$ to its known neighboring APs.

Since power control already satisfies the ordinal potential game condition, we only consider frequency selection. We also assume $|\mathcal{N}|>>|\mathcal{N}_i|$, so that a decrease in power at $i$ results in an increase in the sum of utility functions over all $j\in\mathcal{N}\backslash\mathcal{N}_i$.

Consider first a decrease in $u_i$ due to an increase in the first or last sum of (\ref{eq:new-utility}). Then we can set the second sum constant, and the game emits an exact potential function as in the case where all APs are known. Thus, the interesting case is when there is a decrease due to an increase in the second sum.

To show this part we investigate the utility functions of all $j\in\mathcal{N}_i$ as a result of a change in $u_i$. Let $l\in\mathcal{N}_i$ be the APs which transmit on the same channel as $i$ before $i$ changes channel. Then
\begin{equation}
u_{l\in\mathcal{N}_i}(s_i) = -p_ig_{il} - C_l-p_lg_{li}
\end{equation}
where $C_l$ is some constant which will not change due to $i$ changing its frequency. Similarly we denote $p\in\mathcal{N}_i$ as the APs which transmit on the same channel as $i$ after $i$ changes channel.
\begin{equation}
u_{p\in\mathcal{N}_i}(s_i) = C_p
\end{equation}
Finally we have that $u_i(s_i)$ is given as
\begin{equation}
u_i(s_i) = -\sum_{l\in\mathcal{N}_i}p_jg_{ji} - I_i - \sum_{l\in\mathcal{N}_i}p_ig_{ij}
\end{equation}
where $I_i$ is the term due to all $j\in\mathcal{N}\backslash\mathcal{N}_i$. Now assume that if $i$ changes its strategy from $s_i$ to $s'_i$ that the term $I_i$ decreases by $\epsilon$. I.e. $I'_i = I_i-\epsilon$. By such a decrease in interference at $i$, $i$ can decrease its power if it changes its strategy to $s'_i$. The fractional decrease is, neglecting thermal noise,
\begin{equation}
p'_i = \frac{\sum_{l\in\mathcal{N}_i}p_lg_{li}+I_i-\epsilon}{\sum{l\in\mathcal{N}_i}p_lg_{li}+I_i}p_i
\end{equation}
The new utilities then become
\begin{align}
u_{l\in\mathcal{N}_i}(s'_i) &= C_l \\
u_i(s'_i) &= \sum_{p\in\mathcal{N}_i}p_pg_{pi} - I_i + \epsilon \nonumber \\&- \sum_{p\in\mathcal{N}_i}\frac{\sum_{l\in\mathcal{N}_i}p_lg_{li}+I_i-\epsilon}{\sum_{l\in\mathcal{N}_i}p_lg_{li}+I_i}p_ig_{ip} \\
u_{p\in\mathcal{N}_i}(s'_i) &= -\frac{\sum_{l\in\mathcal{N}_i}p_lg_{li}+I_i-\epsilon}{\sum_{l\in\mathcal{N}_i}p_lg_{li}+I_i}p_ig_{ip} - C_p-p_pg_{pi}
\end{align}
By summation of all utilities before and after we get
\begin{equation}
P(s_i) = -2\sum_{l\in\mathcal{N}_i}p_ig_{ij} - 2\sum_{l\in\mathcal{N}_i}p_lg_{li}-\sum_{l\in\mathcal{N}_i}C_l-\sum_{p\in\mathcal{N}_i}C_p - I_i
\end{equation}
\begin{align}
P(s'_i) &= -2\sum_{p\in\mathcal{N}_i}p_pg_{pi} - 2\sum_{p\in\mathcal{N}_i}p_ig_{ip}\nonumber\\&+2\epsilon\sum_{p\in\mathcal{N}_i}\frac{p_ig_{ip}}{\sum_{l\in\mathcal{N}_i}p_jg_{ji}+I_i}
-\sum_{l\in\mathcal{N}_i}C_l\nonumber\\&-\sum_{p\in\mathcal{N}_i}C_p - I_i+\epsilon
\end{align}

To verify the ordinal potential function we must show that $u_i(s_i,s_{-i})>u_i(s'_i,s_{-i}) \Rightarrow P(s_i,s_{-i})>P(s'_i,s_{-i})$.
\begin{align}
-\sum_{l\in\mathcal{N}_i}p_jg_{ji}-&I_i-\sum_{l\in\mathcal{N}_i}p_ig_{il}>-\sum_{p\in\mathcal{N}_i}p_pg_{pi}-I_i+\epsilon-\nonumber\\&\sum_{p\in\mathcal{N}_i}p_ig_{ip}+\epsilon\sum_{p\in\mathcal{N}_i}\frac{p_ig_{ip}}{\sum_{l\in\mathcal{N}_i}p_jg_{ji}+I_i}
\end{align}
By multiplying each side by $2$ and we see that $P(s_i)>P(s'_i)$.

We see that for this analysis to hold, AP $i$ must have the knowledge of at least one neighboring AP that utilizes each channel.\qed
\section{Proof of Proposition \ref{prop:selfish}}
\label{app:proof-propselfish}
When all APs have the same radius the channel gain between them is also equal (i.e. $g_{ij} = g_{ji}$). By prove convergence by same argument as in the proof of Proposition \ref{prop:resource-conv}, by showing that there exists a utility function which emits a potential function such that the game forms a generalized potential game.

Without any neighbor information, the selfish utility function is given as
\begin{equation}
u_i(k) = \sum_{j\neq i}g_{ji}P_j(k)
\end{equation}
Clearly minimizing $U_i(k)$ is equivalent to minimizing the following, slightly altered utility function
\begin{equation}
u'_i(k) = \sum_{j\neq i}g_{ji}P_j(k)P_i^{\text{nec}}(k).
\end{equation}
Consider a potential function $P(\mathbb{A})$ defined for an allocation $\mathbb{A}$ as
\begin{equation}
P(\mathbf{s}) = \sum_{i=1}^N\sum_{k=1}^{K_i} u'_i(k)
\end{equation}

Now, AP $i$ will change its allocation from channel $k$ to $k'$ given that
\begin{equation}
u_i(k') = \sum_{j\neq i}g_{ji}P_j(k')<\sum_{j\neq i}g_{ji}P_j(k) = u_i(k)
\end{equation}
By multiplying each side by $P_i$ we have
\begin{equation*}
u'_i(k') =\sum_{j\neq i}g_{ji}P_j(k')P_i^{\text{nec}}(k')<\sum_{j\neq i}g_{ji}P_j(k)P_i^{\text{nec}}(k) = u'_{i}(k)
\end{equation*}
Assume before AP $i$ changed from $k$ to $k'$ that AP $l$ transmits on $k$ and AP $m$ transmits on $k'$. Since AP $i$ no longer transmits on $k$, clearly $u_l(k) \text{ (after)}<u_l(k) \text{ (before)}$. %However, the same is not true for AP $m$. Instead we must show that
%\begin{align}
%&\sum_{i=1}^N\sum_{j\neq i} a_i^ka_j^kg_{ji}P_jP_i + \sum_{i=1}^N\sum_{j\neq i} a_i^{k'}a_j^{k'}g_{ji}P_jP_i \text{ (after)} \nonumber \\
%&<\sum_{i=1}^N\sum_{j\neq i} a_i^ka_j^kg_{ji}P_jP_i + \sum_{i=1}^N\sum_{j\neq i} a_i^{k'}a_j^{k'}g_{ji}P_jP_i \text{ (before)}\nonumber
%\end{align}
Define $\Delta u_i$, $\Delta u_l$ and $\Delta u_m$ as
\begin{align*}
\Delta u'_i &= u'_{i}(k) - u'_{i}(k') \\
&= \sum_{j\neq i}g_{ji}P_j(k)P_i^{\text{nec}}(k)-\sum_{j\neq i}g_{ji}P_j(k')P_i^{\text{nec}}(k')>0,
\end{align*}
\begin{align*}
\Delta u_l &= u_l(\text{before $i$ changed}) - u_l(\text{after $i$ changed}) \\
&= g_{il}P_l(k)P_i^{\text{nec}}(k)
\end{align*}
and
\begin{align}
\Delta u_m &= u_m(\text{before $i$ changed}) - u_m(\text{after $i$ changed}) \nonumber \\
&= -g_{im}P_m(k)P_i^{\text{nec}}(k) \nonumber
\end{align}
$\Delta u_l$ corresponds to the users that have gained by user $i$'s change from channel $k$ to $k'$ in terms of increased SINR. $\Delta u_m$ corresponds to the users that have lost by user $i$'s change from channel $k$ to $k'$ in terms of decreased SINR. %Due to reciprocal channel gains we see that
The total change over all affected users is thus
\begin{align}
&\Delta P(\mathbf{s},\mathbf{s}') = \Delta u_i + \sum_{l\neq i} \Delta u_l + \sum_{m\neq i} \Delta u_m \nonumber \\
&=\sum_{j\neq i}g_{ji}P_j(k)P_i^{\text{nec}}(k)-\sum_{j\neq i}g_{ji}P_j(k')P_i^{\text{nec}}(k') \nonumber \\
&+ \sum_{j\neq i}g_{ij}P_j(k)P_i^{\text{nec}}(k)-\sum_{j\neq i}g_{ij}P_j(k')P_i^{\text{nec}}(k') \nonumber \\
&=^{(a)} 2\biggl(\sum_{j\neq i}a_{j}^{k}g_{ji}P_jP_i-\sum_{j\neq i}a_{j}^{k'}g_{ji}P_jP_i\biggr)>0
\end{align}
where the equality $(a)$ holds as long as $g_{ij} = g_{ji}$. Thus if $u_i(k)>u_i(k')$ $\Rightarrow$ $P(\mathbf{s})>P(\mathbf{s'})$ and this is thus a generalized ordinal potential game.\qed
%\begin{equation*}
%u_{i}^{k} = \sum_{l\neq i} \Delta u_l \hspace{0.5cm}\text{and}\hspace{0.5cm}u_i^{k'} = \sum_{m\neq i} \Delta u_m.
%\end{equation*}
%Thus we have that
%\begin{equation}
%\Delta u_i + \sum_{l\neq i} \Delta u_l + \sum_{m\neq i} \Delta u_m
%\end{equation}

%\begin{equation}
%\delta_i+\sum_{l\neq i}\delta_l -\sum_{m\neq i} \delta_m>0
%\end{equation}
%and that $F(\mathbb{A})$ strictly decreases after a change and will thus converge.\qed

\bibliographystyle{ieeetr}
\bibliography{../../../bib,../../../library}{}

\end{document}